\documentclass[11pt]{article}
\usepackage[margin=1in]{geometry}
\usepackage[T1]{fontenc}
\usepackage[utf8]{inputenc}
\usepackage{lmodern,microtype}
\usepackage{amsmath,amssymb,mathtools,amsthm}
\usepackage{graphicx,booktabs,array,tabularx}
\usepackage{tikz}
\usetikzlibrary{arrows.meta}
\usepackage{xcolor,algorithm}
\usepackage[noend]{algpseudocode}
\usepackage[numbers,sort&compress]{natbib}
\usepackage{xurl,hyperref}
\usepackage[nameinlink,capitalise]{cleveref}
\hypersetup{colorlinks=true,linkcolor=blue!45!black,citecolor=green!35!black,urlcolor=blue!55!black,
 pdftitle={SparseDesign: Scaling Exact Coding-Sequence Design},pdfauthor={Hao Lin and Jingjin Yu},
 pdfsubject={Technical report},pdfkeywords={RNA design, codon optimization, candidate sparsification}}
\newtheorem{theorem}{Theorem}

\newcommand{\turner}{Turner~2004}
\newcommand{\sparsedesign}{\textsc{SparseDesign}}

\newcommand{\cai}{\mathrm{CAI}}

\newcommand{\code}[1]{\texttt{\detokenize{#1}}}

\newcommand{\CorpusRows}{7{,}600}
\newcommand{\RetentionZero}{3.53}
\newcommand{\RetentionFour}{2.15}
\newcommand{\RetentionReductionZero}{28.3}
\newcommand{\RetentionReductionFour}{46.4}
\newcommand{\DensityZero}{0.014341}
\newcommand{\DensityFour}{0.008757}
\newcommand{\RetentionMaxZero}{43.80}
\newcommand{\RetentionMaxFour}{37.23}

\newcommand{\DensityClusterCount}{1565}
\newcommand{\DensitySingletonClusters}{1328}
\newcommand{\DensityLargestCluster}{20}

\newcommand{\DensityBootstrapReplicates}{10000}

\newcommand{\DensityZeroBeta}{1.527}
\newcommand{\DensityZeroRSquared}{0.993}
\newcommand{\DensityZeroClusterBeta}{1.519}
\newcommand{\DensityZeroCalibratedClusterBeta}{1.525}

\newcommand{\DensityZeroClusterCILow}{1.520}
\newcommand{\DensityZeroClusterCIHigh}{1.533}

\newcommand{\DensityZeroHingeReductionPercent}{24.1}

\newcommand{\DensityFourBeta}{1.487}
\newcommand{\DensityFourRSquared}{0.991}
\newcommand{\DensityFourClusterBeta}{1.477}
\newcommand{\DensityFourCalibratedClusterBeta}{1.485}

\newcommand{\DensityFourClusterCILow}{1.480}
\newcommand{\DensityFourClusterCIHigh}{1.495}

\newcommand{\DensityFourHingeReductionPercent}{27.5}
\newcommand{\DensityHostBetaMinimum}{1.377}
\newcommand{\DensityHostBetaMaximum}{1.497}
\newcommand{\DensityHostMedianMinimum}{0.005356}
\newcommand{\DensityHostMedianMaximum}{0.009081}

\newcommand{\ValidationDesignCases}{130}
\newcommand{\ValidationEmittedDesigns}{128}

\newcommand{\HistAblationPairs}{24}

\newcommand{\HistPackedSolveIncreasePercent}{8.18}
\newcommand{\HistPackedRSSReductionPercent}{28.21}
\newcommand{\HistPackedRSSGiB}{14.417}

\newcommand{\NativePanelSize}{12}
\newcommand{\NativeFeasibilityTotal}{96}
\newcommand{\NativeFeasibilitySuccessful}{42}
\newcommand{\NativeFeasibilityTimeouts}{8}
\newcommand{\NativeFeasibilitySkipped}{46}
\newcommand{\NativeFreshRepetitions}{210}
\newcommand{\NativeSuccessfulExecutions}{252}
\newcommand{\NativeDistinctOutputs}{36}
\newcommand{\NativeDistinctRNAs}{30}
\newcommand{\NativeRepetitionsPerCondition}{5}
\newcommand{\NativeMatchedConditions}{4}
\newcommand{\NativeRatioMin}{4.71}
\newcommand{\NativeRatioMax}{6.04}
\newcommand{\NativeMatchedLengths}{82, 95, 255, 310}
\newcommand{\NativeMatchedRatios}{5.08, 4.71, 5.51, 6.04}
\newcommand{\NativeOwnRNARefoldGap}{1.40}
\newcommand{\NativeWallCapSeconds}{240}
\newcommand{\NativeAddressCapGiB}{32}
\newcommand{\NativeAuditedArrays}{14}
\newcommand{\NativeAuditedScalars}{7}
\newcommand{\NativeAuditedMotifLists}{3}

\newcommand{\ControlPlannedRuns}{720}
\newcommand{\ControlSuccessfulRuns}{720}
\newcommand{\ControlTimeoutRuns}{0}

\newcommand{\ControlCompleteConditions}{48}
\newcommand{\ControlPlannedConditions}{48}
\newcommand{\ControlPruningMinimum}{1.055}
\newcommand{\ControlPruningMaximum}{1.114}
\newcommand{\ControlReplacementMinimum}{1.321}
\newcommand{\ControlReplacementMaximum}{1.424}
\newcommand{\ProfileEligibleReductionZero}{26.30}
\newcommand{\ProfileEligibleReductionFour}{41.54}
\newcommand{\ControlMaximumObjectiveDeltaUnits}{5.82\times10^{-11}}
\newcommand{\ScalingSixteenMinimum}{10.00}
\newcommand{\ScalingSixteenMaximum}{12.77}
\newcommand{\WorkstationPackedWallMedian}{236.54}
\newcommand{\WorkstationPackedRSSMedian}{14.43}
\newcommand{\WorkstationSquareWallMedian}{237.39}
\newcommand{\WorkstationSquareRSSMedian}{20.09}
\newcommand{\WorkstationPackedOverSquareWallMedian}{0.996}
\newcommand{\WorkstationPackedOverSquareWallMin}{0.985}
\newcommand{\WorkstationPackedOverSquareWallMax}{1.007}

\newcommand{\WorkstationPackedMemoryReductionPercent}{28.20}

\newcommand{\CommodityPackedWallMedian}{126.42}
\newcommand{\CommodityPackedRSSMedian}{14.43}
\newcommand{\CommoditySquareWallMedian}{128.26}
\newcommand{\CommoditySquareRSSMedian}{20.10}

\newcommand{\CLICheckedExecutions}{70}
\newcommand{\CLIDistinctOutputs}{5}
\newcommand{\CLIMaximumDPResidual}{1.75\times10^{-10}}
\newcommand{\ProfileThreadPairs}{72}

\newcommand{\ControlledCompactRows}{%
0 & 1 & 12/12 & 1.055 [1.051, 1.063] & 1.321 [1.312, 1.329] & 1.249 [1.242, 1.256] \\
0 & 16 & 12/12 & 1.091 [1.068, 1.107] & 1.399 [1.366, 1.416] & 1.276 [1.266, 1.290] \\
4 & 1 & 12/12 & 1.067 [1.057, 1.071] & 1.326 [1.322, 1.342] & 1.246 [1.241, 1.261] \\
4 & 16 & 12/12 & 1.114 [1.100, 1.136] & 1.424 [1.399, 1.438] & 1.270 [1.250, 1.290] \\
}

\newcommand{\ScalingCompactRows}{%
Q54TT4 (925) & 0 & 127.54 & 2.23 & 4.20 & 7.36 & 12.77 \\
Q54TT4 (925) & 4 & 120.56 & 2.12 & 4.07 & 7.25 & 12.13 \\
Q61879 (1,976) & 0 & 584.36 & 1.96 & 3.72 & 6.36 & 10.26 \\
Q61879 (1,976) & 4 & 583.29 & 1.98 & 3.73 & 6.39 & 10.00 \\
}

\title{\sparsedesign: Scaling Exact Coding-Sequence Design}
\author{Hao Lin\\{\small Department of Mechanical and}\\{\small Aerospace Engineering}\\{\small Rutgers University, Piscataway, NJ, USA}
\and Jingjin Yu\thanks{Correspondence: \href{mailto:jingjin.yu@cs.rutgers.edu}{jingjin.yu@cs.rutgers.edu}.}\\{\small Department of Computer Science}\\{\small Rutgers University, Piscataway, NJ, USA}}
\date{Technical report --- 26 September 2026}
\begin{document}
\maketitle
\begin{abstract}
Exact optimization of synonymous coding sequences under a joint folding-energy and codon-usage
objective is limited by expensive dynamic-programming splits and large working sets.
\sparsedesign{} applies candidate sparsification to the multiloop recurrence of a
\turner{} dangle-0 solver over a weighted codon automaton. A direct branch is retained only
when it strictly improves on every partitionable or endpoint-unpaired realization of the
same endpoint states. We prove equivalence to the dense recurrence in real arithmetic, under
an explicit scalar branch-interface assumption. With $N$ automaton states, edge set $E$ and
$Z$ retained candidates, multiloop work is $O(N^2+N|E|+NZ)$; worst-case time remains cubic
for bounded-width automata and total memory remains quadratic. Endpoint ownership permits
parallel candidate construction without locks. While synthetic stress families can benefit
little from sparsification and exhibit near-quadratic candidate growth, natural proteins show
substantial candidate-count reductions. In our \CorpusRows{}-task campaign, the 2,000-protein
human-table panel has median retention of only \RetentionZero\% at $\lambda=0$ and
\RetentionFour\% at $\lambda=4$, corresponding to approximately \RetentionReductionZero-fold and \RetentionReductionFour-fold
reductions relative to all feasible direct intervals. The primary performance experiments
use an AMD EPYC 7313 server. For human Dp427c (11,031 nt, $\lambda=0$),
16-thread packed \sparsedesign{} achieves five-run medians of
\WorkstationPackedWallMedian{} seconds wall-clock time and \WorkstationPackedRSSMedian{} GiB peak RSS.
Compared with the single-thread local dense LinearDesign fork on the same server
(4,912 seconds, 402.10 GiB RSS), this gives a 20.8-fold wall-clock speedup and a
27.9-fold peak-memory reduction. On a Core i9-14900KF commodity PC with 64 GiB RAM,
the same input, layout and thread count achieve \CommodityPackedWallMedian{} seconds
and \CommodityPackedRSSMedian{} GiB RSS.
\end{abstract}

\section{Introduction}
\label{sec:intro}
A protein generally admits exponentially many synonymous coding sequences. Jointly choosing
codons and a compatible secondary structure permits exact optimization of a folding-energy
and codon-usage objective, but the practical cost grows sharply with sequence length.
LinearDesign expresses the synonymous design space as a weighted deterministic finite automaton
(DFA) and solves the resulting lattice-parsing problem \citep{Zhang2023LinearDesign}.
For exact loop-based folding, the familiar cubic-time and quadratic-memory costs make long
coding regions and repeated objective sweeps demanding.

The multiloop recurrence repeatedly combines a branch-containing prefix with a final closed
branch. Many direct branches have an alternative realization with the same endpoints and no
greater cost. Fixed-sequence sparse folding exploits this dominance relation
\citep{Backofen2011,Will2016}. In coding-sequence design, the endpoints are automaton states:
branching and merging paths must remain compatible, and codon costs must be charged once.
Moreover, whether a scalar branch cost is sufficient depends on what the consuming loop can
observe about its boundary nucleotides.

This report presents \sparsedesign{}, which specializes this principle to weighted codon
automata under \turner{} dangle-0 energetics. It provides a state-consistent recurrence,
dense-equivalence proof and parallel implementation. An otherwise identical scalar dense
control isolates pruning. Correctness, density and timing experiments characterize practical
gains; synthetic and constrained inputs demonstrate the limits of empirical sparsity.

For human Dp427c (11,031 nt, $\lambda=0$), 16-thread packed \sparsedesign{} achieves
a 20.8-fold wall-clock speedup and a 27.9-fold peak-memory reduction over the single-thread
local dense LinearDesign fork on the same EPYC 7313 server (\cref{sec:performance}).
These ratios include implementation and thread-count differences.

\section{Relation to prior work}
\label{sec:related}
\paragraph{Exact coding-sequence design.}
Polynomial-time codon choice predates modern lattice formulations \citep{Cohen2003}.
CDSfold incorporates amino-acid constraints into a loop-energy dynamic program
\citep{Terai2016}. LinearDesign adds a weighted codon-DFA formulation for the joint energy/CAI
objective and a separate beam-pruned approximation \citep{Zhang2023LinearDesign,Huang2019}.
DERNA and LinearCDSfold provide alternative exact formulations and supported weighted-sum
trade-off search \citep{Gu2024DERNA,Ju2025LinearCDSfold,Liu2026LinearCDSfold}.
A fair comparison must match the energy model, exact versus approximate mode, codon-table
normalization, $\lambda$ units and stop-codon convention; matching an MFE alone is insufficient
when the codon penalty also enters the objective.

\paragraph{Candidate sparsification.}
Accessible-motif and sparse-folding methods established the general reassociation argument
\citep{Wexler2007,Backofen2011}. SparseMFEFold is the closest fixed-sequence predecessor:
its strict candidate test and rightmost normalization yield the same algebraic
$O(n^2+nZ)$ bound under a no-dangle Turner model \citep{Will2016}.
SparseRNAfolD and memerna extend sparse folding to richer boundary interactions
\citep{Gray2024,Courtney2025}. Cho, Chang and Lu study sparse CDS design on a branching codon
representation under a base-pair energy model \citep{Cho2026SparseCDS}. Their model and sparsity
parameters differ from the loop-based state-interval quantity $Z$ here. We therefore claim
neither the first sparse CDS algorithm nor the first sparse algorithm on a codon graph.
Our theorem addresses the loop-based scalar interface and compatible weighted-state joins.
Parallel fixed-sequence folding is also established, for example in GTfold
\citep{Swenson2012GTfold}.

\paragraph{Model scope and recent directions.}
\citet{Ward2025Comparison} provide a correctness benchmark and explicitly propose combining
codon graphs with Aho--Corasick automata for forbidden-motif avoidance. We credit that proposal
and use their published regression fixtures; the contribution here is its evaluated
integration with the sparse solver.
The April 2026 preprint of \citet{Fornace2026Tensor} develops tensor-based codon-constrained
Boltzmann sampling, marginals and joint sequence/structure/configuration minimization,
with CPU/GPU parallelism. Its minimization is over a microstate triplet, whereas finding the
sequence with lowest ensemble free energy retains a sum over its structures. Its tensor
energy model and reported inference/sampling timings do not provide a matched Turner-d0
minimization speed comparison.
The available ChimeraFold workshop manuscript describes an exact parallel codon-graph DP
without sparsification; its software is described as proprietary \citep{Anonymous2026ChimeraFold}.
Its acceptance and public release status could not be verified from primary sources.

EnsembleDesign uses a probabilistic lattice and continuous relaxation for ensemble-based
design \citep{Dai2025EnsembleDesign}; JAX-RNAfold supplies scalable differentiable folding
\citep{Krueger2025JAX}. RNARL, published as an early accepted journal article in July 2026,
combines reinforcement learning with multi-objective sequence generation \citep{Lin2026RNARL}.
The May 2026 GoForth preprint studies conditional generation under combined structure,
sequence and coding constraints \citep{Lindsey2026GoForth}.
CoDOn, an August 2026 early accepted article, uses NSGA-II for multiple codon and local-context
criteria with user constraints \citep{Moon2026CoDOn}.
These methods broaden the objective and search-method landscape. Our dense-equivalence theorem
addresses the explicitly stated MFE--CAI objective, not ensemble or learned-objective
optimality. The literature search and reference audit through 26 September 2026 support this
model-specific positioning without a priority claim for sparse CDS optimization.

\section{Problem and model}
\label{sec:model}
Let protein input $p$ specify $A$ codon positions and let $n=3A$. Ordinarily $A$ is the amino-acid
count; an explicit stop symbol adds one codon position. No stop is appended
implicitly. For a codon table with relative adaptiveness
$w(c)=f(c)/\max_{c'\sim c}f(c')\in(0,1]$, define
$C(r)=-\sum_{c\in r}\log w(c)=-A\log\cai(r)$. The objective is
\begin{equation}
 F^*(p,\lambda)=\min_{r\in\operatorname{Syn}(p)}\;
 \min_{S\in\mathcal S_{3,30}(r)}
 \left\{E_{\mathrm{T04,d0}}(r,S)+\lambda C(r)\right\},\qquad \lambda\geq0.
 \label{eq:objective}
\end{equation}
Here $\mathcal S_{3,30}$ contains pseudoknot-free structures with AU, CG and GU pairs,
a minimum hairpin size of three, and at most 30 unpaired nucleotides in each bulge/internal
loop. Energies use \turner{} parameters \citep{Mathews2004} at $37^\circ$C. Our model enables
special tri-, tetra- and hexaloops and lonely pairs, but excludes G-quadruplexes,
dangling-end contributions and coaxial stacking.
Multiloop costs are affine in the number of incident helices and unpaired bases, with
additional terminal-pair penalties; the closing and branch terms are accounted for explicitly
by the typed interface in \cref{eq:branch}.
The cap is a restriction of the optimization domain, not a restriction on what a separate
structure-energy evaluator can score. CAI follows \citet{SharpLi1987}.
The CLI uses kcal/mol for $\lambda$ and multiplies it by 100 before calling the C++ core, which
combines double-precision codon costs with integer energies in units of $0.01$ kcal/mol.

A layered DFA is $\mathcal A=(Q,E,s,t,\rho,\ell,q)$, where $\rho$ is nucleotide position and
each edge advances one layer, emits $\ell(e)\in\{A,C,G,U\}$, and carries an unscaled codon
penalty $q(e)$. The complete $-\log w(c)$ is assigned to the codon's third edge; the other two
edges have zero weight. Thus source-to-sink path sums reproduce $C(r)$ exactly as an algebraic
objective. We write $N=|Q|$ and $g(a,b)=\rho(b)-\rho(a)$.
Intervals cover paths between specific states, not just nucleotide positions. Distinct states
at the same layer have no empty path between them. The unmodified codon construction has
bounded width and degree, giving $N=\Theta(n)$ and $|E|=O(N)$.

\begin{figure}[tb]
\centering
\begin{tikzpicture}[>=Stealth,vertex/.style={circle,draw,inner sep=2.5pt},font=\small]
\node[vertex] (a) at (0,0) {$a$};
\node[vertex] (c1) at (3,0.45) {$c_1$};
\node[vertex] (c2) at (3,-0.45) {$c_2$};
\node[vertex] (b) at (6,0) {$b$};
\draw[->,thick,blue!65!black] (a) -- node[above,sloped] {prefix path} (c1);
\draw[->,thick,green!40!black] (c2) -- node[below,sloped] {closed branch} (b);
\draw[dashed,gray] (3,-0.9) -- (3,1.05);
\node at (3,1.3) {$\rho(c_1)=\rho(c_2)$};
\node at (3,-1.18) {$c_1\ne c_2$: concatenation is infeasible};
\end{tikzpicture}
\caption{Equal nucleotide positions do not identify automaton states. The illustrated
prefix and branch cannot be concatenated; every DP split must use the same joining
state. Edges shown schematically may stand for multi-edge paths.}
\label{fig:state-join}
\end{figure}
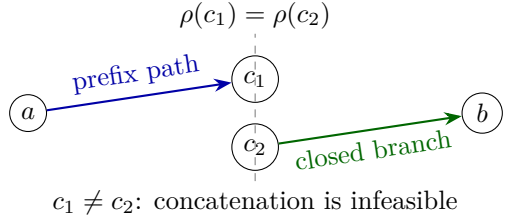

Following the codon-graph/Aho--Corasick proposal of \citet{Ward2025Comparison}, intersecting
the codon DFA with a deterministic forbidden-motif automaton preserves this layered path model.
Matcher history crosses codon boundaries, and transitions entering forbidden states are removed.
A fixed finite motif inventory preserves bounded-width asymptotics, but increasing the reachable
matcher state count $K$ can multiply lattice size by $K$ and the number of interval slots by
$K^2$. The observed state growth is measured in \cref{sec:constraints}.

\section{Candidate-sparse multiloops}
\label{sec:algorithm}
\subsection{Complete dense and sparse recurrences}
A typed closed value $V_\tau(a,b)$ includes the optimal path and loop energy of an interval
closed by oriented pair type $\tau$. Hairpin, stack/bulge/internal-loop and multiloop
closures produce these values using compatible endpoint states and the model restrictions
in \cref{sec:model}. A branch attached to a multiloop has cost
\begin{equation}
 B_\tau(a,b)=V_\tau(a,b)+\beta_\tau,\qquad
 B(a,b)=\min_\tau B_\tau(a,b),
 \label{eq:branch}
\end{equation}
where $\beta_\tau$ includes its multiloop branch and terminal-pair penalties. Typed values
are retained for external and two-loop consumers. The scalar minimum is formed only at
the multiloop interface.

Let $u(e)=\lambda q(e)+c_{\rm ML}$ combine an unpaired edge's codon penalty with the per-base
multiloop term. $D_1(a,b)$ and $D_2(a,b)$ denote fragments containing at least one and at least
two top-level branches. Their right-normal dense recurrences decompose fragments from the right:
\begin{align}
 D_2(a,b)&=\min\left\{
  \min_{e=(b',b)}[D_2(a,b')+u(e)],\;
  \min_{\rho(a)<\rho(c)<\rho(b)}[D_1(a,c)+B(c,b)]\right\},
  \label{eq:dense2}\\
 P_D(a,b)&=\min\left\{
  D_2(a,b),\;
  \min_{e=(a,a')}[u(e)+D_1(a',b)],\;
  \min_{e=(b',b)}[D_1(a,b')+u(e)]\right\},
  \label{eq:partitionable}\\
 D_1(a,b)&=\min\{B(a,b),P_D(a,b)\}.
 \label{eq:dense1}
\end{align}
All minima range over actual edges or concrete intermediate states. An empty minimum and
an infeasible interval have value $+\infty$. One- and two-branch values are $+\infty$ on
zero-span intervals, including distinct same-layer states; a direct branch has strictly
positive span. Arithmetic is over finite real costs and the infeasibility sentinel.

The alternatives in $P_D$ realize the same endpoints without using $B(a,b)$ as the sole
branch. Retain a direct branch only if it is strictly better:
\begin{equation}
 \mathcal C_b=\{c:B(c,b)<P_D(c,b)\},\qquad
 Z=\sum_b|\mathcal C_b|.
 \label{eq:candidates}
\end{equation}
Replace only the dense split scan by
\begin{equation}
 S_2(a,b)=\min\left\{
  \min_{e=(b',b)}[S_2(a,b')+u(e)],\;
  \min_{\substack{c\in\mathcal C_b\\\rho(c)>\rho(a)}}[S_1(a,c)+B(c,b)]\right\}.
 \label{eq:sparse2}
\end{equation}
Define $P_S$ by \cref{eq:partitionable} with $S$ replacing $D$, and set
$S_1=\min\{B,P_S\}$. The online implementation uses $B<P_S$; the proof below establishes
that this yields the same candidate decision as the dense test before any larger interval uses it.
A candidate list stores only the start state and scalar cost. Direct terms are always available
to $S_1$, including terms that are not retained as later split candidates.

\subsection{Dense equivalence}
\begin{theorem}[Relative exactness]
\label{thm:exact}
Assume (i) a finite layered acyclic automaton with strictly advancing edges and branches;
(ii) additive costs and joins only at identical endpoint states; (iii) a multiloop parent
consumes type $\tau$ through $x+B_\tau(a,b)$ with $x$ independent of $\tau$; and
(iv) the surrounding dense Turner recurrences preserve a compatible path and charge each
edge exactly once. In real arithmetic, the sparse and dense multiloop values agree at every
interval. If each typed closed cell depends only on smaller-span closed/multiloop cells,
the complete sparse solver has the same optimum as the dense solver under \cref{eq:objective}.
\end{theorem}
\begin{proof}
First hold the direct costs $B$ fixed. Regard a finite $B(a,b)$ as an atomic arc on the layered
graph. A multiloop fragment is a path of branch arcs and unpaired edges; only top-level arcs
count toward its branch count. Decomposing a path with at least two branches at its final
unpaired edge or rightmost branch gives \cref{eq:dense2}. A path with at least one branch is
either a sole direct arc, has at least two branches, or has an unpaired edge at an endpoint.
Peeling that edge gives \cref{eq:partitionable,eq:dense1}. This also explains equivalence
with a left-normal dense enumeration: both enumerate the same path family.

Take a dense path with at least two branches and rightmost branch $B(c,d)$. If $(c,d)$ is
not retained, $P_D(c,d)\leq B(c,d)<\infty$. The finite path family therefore has an attaining
alternative. Replace the arc by that alternative. The new path has the
same endpoint states, contains at least one branch in place of the old arc, and has no larger
cost. Hence the requirement of at least two branches and automaton feasibility are preserved.
An endpoint-unpaired alternative shortens the interval occupied by its rightmost branch;
a two-branch alternative contains at least two ordered positive-span arcs, so its rightmost
branch also has strictly smaller span. Repetition therefore terminates at a path whose
rightmost branch is a candidate. A tie causes no difficulty and need not be stored.

Induct on interval span, evaluating $S_2$ before $S_1$. In a normalized dense optimum, the
branch-containing prefix before the final candidate has smaller span. Its value equals
$S_1$ by induction; any unpaired suffix is added by repeated edge-extension terms in
\cref{eq:sparse2}. Thus $S_2\leq D_2$. Conversely, every sparse realization is admitted by
the dense recurrence, so $S_2\geq D_2$. Equality of $S_1$ follows from its unchanged direct
term, the equality just established for $S_2$, and equal smaller-span extension values.
Consequently $P_S=P_D$, and the online candidate test is the dense test. In a split for
$(a,b)$, the guard $\rho(c)>\rho(a)$ makes the consumed candidate $(c,b)$ strictly
smaller-span than $(a,b)$. Its decision is therefore complete before consumption, avoiding
same-span circularity.

Assumption (iii) justifies the scalar collapse:
$\min_\tau[x+B_\tau]=x+\min_\tau B_\tau=x+B$.
Under dangle-0, the type-dependent multiloop contribution is already inside $\beta_\tau$;
no flanking or neighboring-branch term changes the minimizing type. This does not collapse
$V_\tau$ for consumers whose interfaces still depend on type.

Finally apply a joint span induction to Turner cells. A hairpin closure is local to its
compatible path, a two-loop uses a smaller inner closed span, and a multiloop uses a smaller
$D_2$ span after consuming its boundary pair. Equal smaller cells therefore produce equal
$V_\tau$ and $B$ at the current span; the preceding multiloop argument then produces equal
$D_1,D_2$ and candidate decisions. The unchanged external recurrence completes the induction.
\end{proof}

The theorem proves equivalence to the stated dense model, rather than rederiving the entire
dense Turner grammar. Dangles or coaxial stacking can make the preferred branch type depend
on its consumer and require a richer candidate interface. The proof also does not establish
bit-identical floating-point execution: the implementation uses doubles, and its numerical
and reconstruction checks are reported separately.

\subsection{Work and storage}
There are $O(N^2)$ interval cells. Incident-edge extensions cost $O(N|E|)$ in aggregate,
and each retained candidate is visited for at most $O(N)$ left endpoints. Thus multiloop
work is $O(N^2+N|E|+NZ)$, or $O(n^2+nZ)$ for bounded-width codon automata.
The candidate index occupies $O(N+Z)$ storage. It does not replace the quadratic closed,
gap, codon-cost and multiloop tables. With bounded width and fixed internal-loop cap,
the complete solver remains worst-case cubic in RNA length and quadratic in memory.
The finite synthetic panels in \cref{sec:density} illustrate why empirical sparsity
cannot be promoted to a universal subcubic guarantee.

\section{Parallel evaluation and memory organization}
\label{sec:implementation}
The wavefront orders intervals by nucleotide span. Each span has three phases separated by barriers:
closed/gap values, multiloop values, then external values. The second phase computes $S_2$
before $S_1$ within one interval. For fixed span $g$ and right endpoint $b$, only the start-layer
iteration $\rho(b)-g$ owns that endpoint. It may append to $\mathcal C_b$ without a lock.
The guard $\rho(c)>\rho(a)$ admits only candidate intervals shorter than $g$, so candidates
added for another start state on the current layer are never consumed prematurely.
Dynamic scheduling distributes start-layer iterations; it does not reorder the recurrence
within an interval. The external recurrence is needed only at the final-state column.

\begin{algorithm}[tb]
\caption{Sparse wavefront; each parallel loop ends with a barrier.}
\label{alg:wavefront}
\begin{algorithmic}[1]
\For{$g=1,\ldots,n$}
 \State \textbf{parallel for} $i=0,\ldots,n-g$
 \Statex \hspace{1em} fill codon/gap and typed closed cells for layers $(i,i+g)$
 \State \textbf{parallel for} $i=0,\ldots,n-g$
 \Statex \hspace{1em} for each endpoint pair: fill $S_2$, then $P_S,S_1$; retain $B<P_S$
 \State \textbf{parallel for} $a$ at layer $n-g$
 \Statex \hspace{1em} fill the external value ending at final state $t$
\EndFor
\end{algorithmic}
\end{algorithm}

\begin{table}[tb]
\centering\small
\caption{Mechanisms and their scope. Memory effects interact; the table is not a multiplicative
explanation of any observed end-to-end speedup.}
\label{tab:mechanisms}
\begin{tabularx}{\linewidth}{@{}p{0.27\linewidth}X@{}}
\toprule
Mechanism & Effect and limitation\\
\midrule
Scalar candidate lists & Restrict multiloop split work and branch-list storage; retain typed
closed states for other consumers.\\
Contiguous memo arenas & Remove per-cell hash lookup and use compact scalar/vector/flag objects;
heap payload and interval count remain substantial.\\
Deferred traceback & Reconstruct decisions from completed costs and lazily materialize needed
branch rows; no complete decision-table duplication.\\
Final-column external DP & Store and fill only the external values ending at the final state.\\
Forward-packed option & Omit reverse-layer interval slots; preserve same-layer base cases and
trade indexing cost against working set.\\
\bottomrule
\end{tabularx}
\end{table}

The square arena is the default; the named packed build is an explicit memory-saving alternative.
On the measured 64-bit implementation, two vector objects, three scalar doubles and two flags use
about 74 bytes per interval slot before heap payload. A candidate slot occupies 16 bytes plus
one 24-byte vector header per right endpoint. Candidate-container savings are consequently distinct
from whole-process memory savings against a historical implementation.

Traceback replays dense decisions. Every design is checked against its DFA path and a
separately recomputed structure-energy-plus-path-penalty objective, with full-precision
DP/residual diagnostics. Separate profiling instantiations count splits and measure phase
wall times; ordinary solves omit counters and phase clocks. Phase times include instrumentation,
scheduling and barriers and cannot substitute for ordinary-kernel timings. Benchmarks set affinity
externally; the runtime team-cap heuristic neither identifies physical cores nor pins workers.

\section{Evaluation design}
\label{sec:evaluation}
Separate experiments assess optimality, candidate retention, implementation speedups
and memory requirements for complete designs.

\paragraph{Frozen broad panel.}
The campaign selects 2,000 distinct reviewed, nonfragment Swiss-Prot proteins in 12
prespecified strata spanning 10--4,000 amino acids, using a seeded sampling rule from 475,681
unique eligible sequences \citep{UniProt2025}. Each is evaluated with the human table at
$\lambda\in\{0,4\}$. A prespecified 400-protein subset is reused with nine additional codon
distributions at $\lambda=4$ \citep{Nakamura2000}, for \CorpusRows{} tasks in total.
The independent audit verifies all 11 merged-file hashes, exact task membership, uniqueness,
lengths, solver/table identities, status and candidate arithmetic. The merged rows do not contain
every full designed sequence and structure, so this audit is not an independent optimality
certificate for all \CorpusRows{} designs.

\paragraph{Current correctness and stress panel.}
Release checks span published fixtures, natural proteins, objective weights, table precision,
synthetic families and motif constraints. Designs undergo translation, motif-language,
DP/reconstruction and independent ViennaRNA 2.7.2 checks under matched Turner-d0 settings
\citep{Lorenz2011ViennaRNA}; small synonymous spaces are enumerated independently.
\Cref{sec:validation} distinguishes shared and independent implementations.
The package supplies frozen inputs, summaries, sources and external-fixture notices.

\paragraph{Timing boundaries and uncertainty.}
Cost-only kernel measurements include solver allocation, forward evaluation and destruction;
input/DFA construction and traceback are excluded. Complete CLI measurements include those
additional stages. Wall time and whole-process peak resident set size (RSS) are reported
separately, and failures or resource limits remain explicit records. Profiling runs are separate
from ordinary timing runs. The primary timing, native-comparison and validation measurements
use a dual-16-core AMD EPYC 7313 server with 1 TiB RAM and GCC 11.4.0.
Benchmarks ran on a shared server, introducing variability between runs, although
overall load was generally not exceptionally high when runs were attempted.
The Core i9-14900KF Dp427c example is reported separately to demonstrate commodity-PC feasibility.
A candidate-count ratio is neither an instruction-count ratio nor a speedup estimate.
The length-balanced corpus summaries describe the selected panel; they are not Swiss-Prot
population estimates or evidence that homologous proteins are statistically independent.

\section{Correctness and robustness}
\label{sec:validation}
\subsection{Dense, reconstructed and independent objectives}
The release's production-cell test compares every $M_1$, $M_2$ and typed closed value of the
sparse implementation with the scalar right-normal dense and original dense recurrences:
72 configurations and 181,344 intervals per layout. Cases combine targeted multiloops,
seeded random proteins, three objective weights, motif products and zero/nonzero multiloop
unpaired penalties. Dense-right retains every finite direct interval, and sparse candidates
are a subset. The complete suite passes under square and packed layouts, forward-key
assertions, AddressSanitizer and UndefinedBehaviorSanitizer; GCC 14.3 and 15.2 checks also pass.
Leak detection is disabled in the ASan harness, so no leak-sanitizer claim is made.
A separate exact-integer abstract test checks 38,539 cells on 500 generated branching/merging
DAGs. It isolates the recurrence logic from floating-point energy calculations.

\begin{table}[tb]
\centering\small
\caption{External and stress validation on the EPYC server. All \ValidationDesignCases{} prespecified cases pass;
\ValidationEmittedDesigns{} produce designs and two correctly reject an infeasible language.}
\label{tab:validation}
\begin{tabularx}{\linewidth}{@{}Xr@{}}
\toprule
Study & Cases\\
\midrule
Five natural proteins, seven objective weights & 35\\
Published Ward regression fixtures \citep{Ward2025Comparison}, including 950 and 1,000 aa & 8\\
Full-precision versus rounded codon frequencies & 30\\
S and alternating YI families, 32--256 aa, three weights & 24\\
Two natural proteins, five nested motif sets, three weights & 30\\
Overlapping-motif and infeasible-language controls & 3\\
\bottomrule
\end{tabularx}
\end{table}

Every emitted design in \cref{tab:validation} agrees with independent ViennaRNA structure
energy evaluation and fixed-sequence refolding under the stated settings. The panel's largest absolute
DP-versus-separately-recomputed weighted-objective residual is $7.28\times10^{-11}$ in internal
$0.01$-kcal/mol units. The largest public floating-point energy-return discrepancy is
$2.44\times10^{-5}$ kcal/mol. Three dense controls additionally agree on 114 cases restricted to at most 360 nucleotides
and maximum lattice width 16.
The in-house structure evaluator and dense controls share energy functions with the optimizer;
the Python codon-penalty calculation and ViennaRNA checks are separate implementations.
The Ward cases test independently published failure fixtures, not whether current releases of
the tools named in their historical descriptions still exhibit those failures.

Independent synonymous enumeration performs 1,931 feasible evaluations across seven problems
(1,714 distinct RNAs), with all 30 weighted-objective comparisons passing. Nine checks lie at
or immediately around three nonzero support-line breakpoints; $\lambda=0$ provides extensively
tied cases. All 28 loop fixtures pass: 22 special hairpins and six bulge/internal-loop boundary
cases. An exactly structure-constrained ViennaRNA fold is feasible with 30 unpaired bases and
infeasible with 31, whereas its evaluator can score either structure. The in-house
unconstrained fixed-sequence optimum also agrees on each fixture; this does not imply
a hard structure-constraint API in \sparsedesign{}. These complementary checks make the
optimization/evaluation distinction in \cref{eq:objective} explicit.

\subsection{Table precision and model sensitivity}
Rounding normalized human frequencies to two decimals changes one of 15 paired designed
sequences. At $\lambda=4$, that design has a 4.60-kcal/mol higher MFE, compensated by its codon
term; its regret under the original full-precision objective is only 0.00867 kcal/mol.
Near a trade-off boundary, a small objective change can therefore accompany substantial changes
in the two components. Exact-score comparisons should preserve normalization conventions and
should not require identical sequences when optima are tied.

Ensemble-derived accessibility measures are a separate objective. Among 96 within-protein
comparisons across objective weights with distinct RNAs and different d0 MFEs, seven reverse their ordering
under average unpaired probability. No d0/d2-MFE or d0-MFE/d0-ensemble-energy reversal is observed
in this small subset. These pairwise comparisons are dependent and cannot establish general
ranking invariance. They are computational rescoring checks, not expression experiments.

\subsection{Finite-state constraints}
\label{sec:constraints}
The two natural constraint inputs contain 64 and 105 amino acids. At 16 forbidden motifs,
their automata grow from 217 and 362 nodes to 598 and 966, and maximum layer width grows from
two to eight and ten. All returned sequences avoid the complete inventory, and every constrained
objective is at least its unconstrained counterpart. At $\lambda=0$, the penalties are 0.50
and 3.40 kcal/mol; the largest weighted penalty across this panel is 8.717 kcal/mol.
The larger 16-motif case at $\lambda=0$ uses approximately 109.6 MiB before dense checks and takes about
37.7 seconds for construction, solve, traceback and statistics in the instrumented serial
EPYC validation run. These are descriptive measurements, not repeated performance estimates.
The construction demonstrates exact design under finite-state constraints and measures their state-growth cost.

\begin{figure}[tb]
\centering
\includegraphics[width=\linewidth]{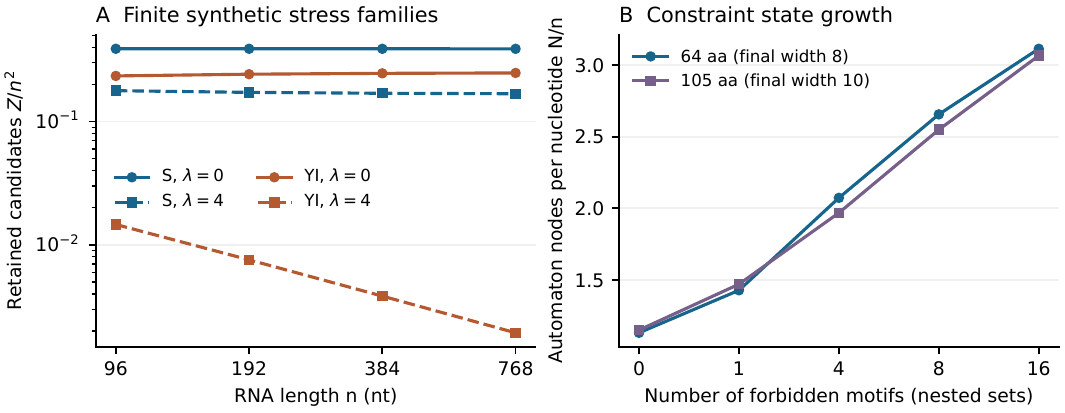}
\caption{Synthetic stress families and finite-state constraints. Left: candidate count divided
by squared RNA length for S and alternating YI proteins; the near-constant ratios at $\lambda=0$
over this finite range indicate near-quadratic growth, without establishing an asymptotic bound.
The YI ratio at $\lambda=4$ decreases over this range. Right: nodes per
nucleotide under nested motif inventories for two natural inputs. Inventory sizes are categorical.
Final maximum widths are eight and ten, versus two initially. All cases use the frozen human
table and the specified Turner-d0 model.}
\label{fig:validation}
\end{figure}

\section{Performance attribution and large-instance evidence}
\label{sec:performance}
\subsection{Three solver configurations}
The typed dense wavefront enumerates left-normal pair-indexed branch rows. The scalar dense
control uses the same right-normal recurrence, pair-type minimum, endpoint lists, wavefront
and memo layout as sparse, but retains \emph{every finite direct interval}. Its only
compile-time difference is the retention predicate. Comparing the scalar configurations isolates
pruning and its resulting list sizes. Comparing typed dense with sparse measures the complete
multiloop replacement, including orientation and pair-type collapse.

\subsection{Primary EPYC comparisons}
\label{sec:current-performance}
The frozen panel selects two proteins per band across six length bands from 80--120 to
1,800--2,000 amino acids: the lowest and highest measured human-table candidate retention
at $\lambda=0$ among the specified model-organism taxids. This purposeful panel is not a
population sample. Five paired repetitions compare three ordinary kernels at $\lambda=0,4$
with one or 16 threads. Cost-call times include solver construction, fill and destruction,
excluding DFA construction and traceback.

The EPYC study uses the frozen solver sources, built with GCC 11.4.0 and
\code{-O3 -flto}. Ordinary runs have 1,800-second wall and 32-GiB address-space
limits; complete Dp427c runs have a 3,600-second wall cap. Pilots are excluded.
An address-space limit does not simulate a physically 32-GiB host.

Of \ControlPlannedRuns{} ordinary runs, \ControlSuccessfulRuns{} succeed and
\ControlTimeoutRuns{} time out, yielding \ControlCompleteConditions/\ControlPlannedConditions{}
complete protein/weight/thread conditions. The maximum spread among successful common-input
kernel objectives is $\ControlMaximumObjectiveDeltaUnits$ internal units.
All three ratios use shared complete three-arm blocks. We take each condition's median paired
ratio, then the geometric mean across conditions with all five blocks within each
weight/thread setting. The 10,000-draw joint
within-condition repetition bootstrap quantifies conditional repeat variation; it neither
resamples proteins nor corrects systematic hardware drift.

\begin{table}[tb]
\centering\small\setlength{\tabcolsep}{3pt}
\caption{Ordinary cost-call ratios: geometric means of condition-wise median paired ratios,
with conditional 95\% bootstrap intervals. $N$ counts complete conditions out of 12.
Every row contains all 12 selected proteins, each with five complete three-arm blocks.}
\label{tab:controlled}
\begin{tabular*}{\linewidth}{@{\extracolsep{\fill}}rrclll@{}}
\toprule
$\lambda$ & $j$ & $N$ & Scalar/sparse & Typed/sparse & Typed/scalar\\
\midrule
\ControlledCompactRows
\bottomrule
\end{tabular*}
\end{table}

\begin{figure}[tb]
\centering
\includegraphics[width=\textwidth]{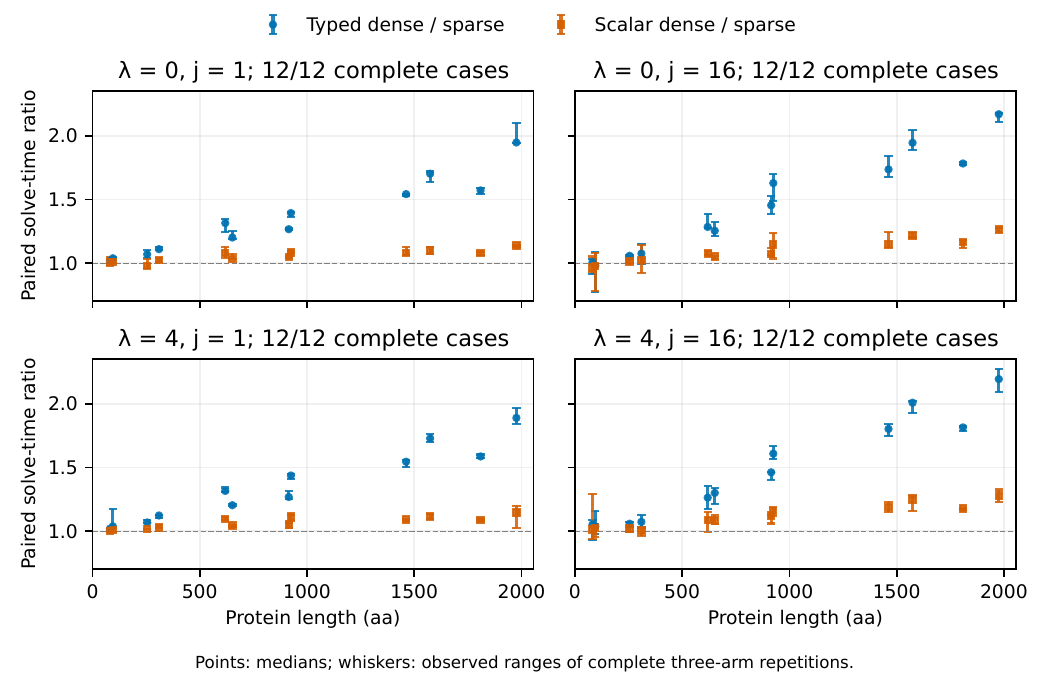}
\caption{Ordinary three-arm comparisons. Points are median paired ratios; whiskers are observed
ranges, not confidence intervals. Every protein/weight/thread condition has five
complete triplets on the EPYC server.}
\label{fig:controlled}
\end{figure}

Across four weight/thread settings, the scalar-dense/sparse aggregate ratio is
\ControlPruningMinimum--\ControlPruningMaximum; typed-dense/sparse gives
\ControlReplacementMinimum--\ControlReplacementMaximum{}
(\cref{tab:controlled,fig:controlled}). These aggregate ranges do not imply uniform
per-protein speedups. Pruning modestly improves whole-call time; the larger replacement
gain includes joint orientation, scalar-collapse and branch-representation changes.

A separate EPYC Q8VIM6 follow-up ($\lambda=4,j=1$) completes all five
three-kernel blocks. Its median paired ratios are scalar/sparse 1.105, typed/sparse 1.597, typed/scalar 1.468.
These descriptive observations are kept separate from the primary estimates.

\paragraph{Mechanism counts and profiling.}
All 144 profile runs succeed. Median scalar-dense/sparse eligible-list-visit ratios are
\ProfileEligibleReductionZero{} at $\lambda=0$ and \ProfileEligibleReductionFour{} at
$\lambda=4$. All \ProfileThreadPairs{} one-versus-16-thread comparisons have identical visit
and eligible counts. Counts measure scanned multiloop work, not candidate inventory or
elapsed speedup. Instrumentation changes execution cost, especially at 16 threads. Phase timers
therefore do not decompose ordinary-kernel speedup.

\paragraph{Complete-command scaling and memory.}
All 60 scaling runs succeed. Ratios of three-run complete-command medians give
\ScalingSixteenMinimum--\ScalingSixteenMaximum{} speedup from one to 16 threads
(\cref{tab:scaling}); this includes startup, DFA construction, design, traceback and output.
\begin{table}[tb]
\centering\small
\caption{Complete-command scaling. Each cell uses three successful runs. Wall time is seconds
at one thread; the other columns give the one-thread/indicated-thread median-time ratio.}
\label{tab:scaling}
\begin{tabular*}{\linewidth}{@{\extracolsep{\fill}}lrrrrrr@{}}
\toprule
Protein (aa) & $\lambda$ & $j=1$ wall & $j=2$ & $j=4$ & $j=8$ & $j=16$\\
\midrule
\ScalingCompactRows
\bottomrule
\end{tabular*}
\end{table}

Ten complete EPYC designs of Dp427c (3,677 aa; 11,031 nt) succeed at $\lambda=0,j=16$.
The square layout's five-run median is \WorkstationSquareWallMedian{} seconds and
\WorkstationSquareRSSMedian{} GiB peak RSS; packed gives \WorkstationPackedWallMedian{}
seconds and \WorkstationPackedRSSMedian{} GiB. Across five counterbalanced pairs, packed/square
wall ratios have median \WorkstationPackedOverSquareWallMedian{} and range
\WorkstationPackedOverSquareWallMin--\WorkstationPackedOverSquareWallMax;
median paired RSS decreases by \WorkstationPackedMemoryReductionPercent\%.
No host swap-in or swap-out pages are recorded during these ten runs.
Independent evaluation validates all \CLICheckedExecutions{} complete-command outputs
(\CLIDistinctOutputs{} distinct RNA/structure pairs) for translation and model admissibility.
Returned-structure energies agree exactly with ViennaRNA 2.7.2 in integer centikcal/mol;
the maximum independently recomputed weighted-DP residual is $\CLIMaximumDPResidual$ internal units.
Evaluation-only scoring allocates no folding matrices and does not independently prove optimality.

\paragraph{Commodity-PC demonstration.}
On a Core i9-14900KF workstation with approximately 64 GiB physical RAM, the same
Dp427c input at $\lambda=0,j=16$ completes all five square and five packed runs.
The packed layout's median complete-command time is \CommodityPackedWallMedian{} seconds
with \CommodityPackedRSSMedian{} GiB peak RSS; square uses \CommoditySquareWallMedian{}
seconds and \CommoditySquareRSSMedian{} GiB. These GCC 14.3.0 builds use a 32-GiB
address-space limit and a 1,800-second wall cap. The 16 threads span performance
and efficiency cores. This demonstrates exact full-length design on a commodity PC;
these measurements are separate from the EPYC performance estimates.

\subsection{Native public-software comparison}
\label{sec:native-baselines}
On the \NativePanelSize-protein panel, we evaluated pinned, unmodified LinearDesign,
exact LinearCDSfold and DERNA \citep{Zhang2023LinearDesign,Liu2026LinearCDSfold,Gu2024DERNA}
against the frozen square-layout \sparsedesign{} CLI at $\lambda\in\{0,4\}$.
Runs used one thread on the AMD EPYC 7313 server, a \NativeWallCapSeconds-s wall cap and
\NativeAddressCapGiB-GiB address-space cap. Complete-CLI time includes parsing, allocation,
traceback, output and native postprocessing; queueing and common rescoring are excluded.
Local builds used GCC~11.4.0 with native \code{-Ofast} (LinearDesign),
\code{-O3} (LinearCDSfold/DERNA), and \code{-O3 -flto} (\sparsedesign);
LinearDesign also loads its upstream-supplied shared library.

Ascending-length feasibility recorded \NativeFeasibilityTotal{} conditions:
\NativeFeasibilitySuccessful{} successes, \NativeFeasibilityTimeouts{} wall timeouts and
\NativeFeasibilitySkipped{} prescribed skips after the first cap per tool/weight.
Skipped lengths remain unattempted; this protocol does not establish maximum feasible lengths.
Each feasible condition received \NativeRepetitionsPerCondition{} fresh repetitions in
randomized order; all \NativeFreshRepetitions{} succeeded. Including feasibility,
\NativeSuccessfulExecutions{} successful executions yielded \NativeDistinctOutputs{}
distinct RNA/structure outputs (\NativeDistinctRNAs{} RNAs), all independently translated
and structurally validated. Digests bind sources, executables, inputs and raw outcomes;
loader/CRLF development pilots are excluded.
Feasibility observations are excluded from median timing estimates; wall and RSS fields
unavailable after termination remain missing in the retained records.

Native models differ. LinearDesign/LinearCDSfold impose 20-unpaired two-loop caps and
additional multiloop restrictions; \sparsedesign/DERNA use 30.
LinearDesign and DERNA parse binary32 frequencies before double normalization;
LinearCDSfold also normalizes in binary32, whereas \sparsedesign{} uses doubles throughout.
DERNA's weighted mode uses $\alpha=1/(1+100\lambda)$ and includes serialization/refolding
in native time. Therefore \cref{fig:native-baselines} reports native measurements;
strict LinearDesign/LinearCDSfold and $\lambda=4$ speed ratios are excluded.

Only \NativeMatchedConditions{} $\lambda=0$ conditions pass the comparison gates:
audited thermodynamic parameters, five validated successes per arm, cap-30 admissibility,
and exact integer-centikcal energy agreement across native reports, common rescoring and arms.
The source audit matches \NativeAuditedArrays{} canonical parameter arrays,
\NativeAuditedScalars{} scalars and the order of \NativeAuditedMotifLists{} special-loop motif lists.
Common scoring uses ViennaRNA~2.7.2 \citep{Lorenz2011ViennaRNA}, \turner, $37^\circ$C,
no dangles and special hairpins. \sparsedesign{} DP costs must equal that integer energy;
realized diagnostics must equal the exact binary64 round trip
$(E_{\rm centi}/100.0)\times100.0$, with no energy tolerance.
For \NativeMatchedLengths{} residues, DERNA/\sparsedesign{} median-time ratios are
\NativeMatchedRatios{}, respectively (range \NativeRatioMin--\NativeRatioMax).
These conditional ratios do not prove recurrence equivalence. Output validation does not
certify common-model optimality: LinearDesign's P08311 $\lambda=4$ structure lies \NativeOwnRNARefoldGap{} kcal/mol above
common fixed-RNA refolding.

\begin{figure}[tb]
\centering
\includegraphics[width=\textwidth]{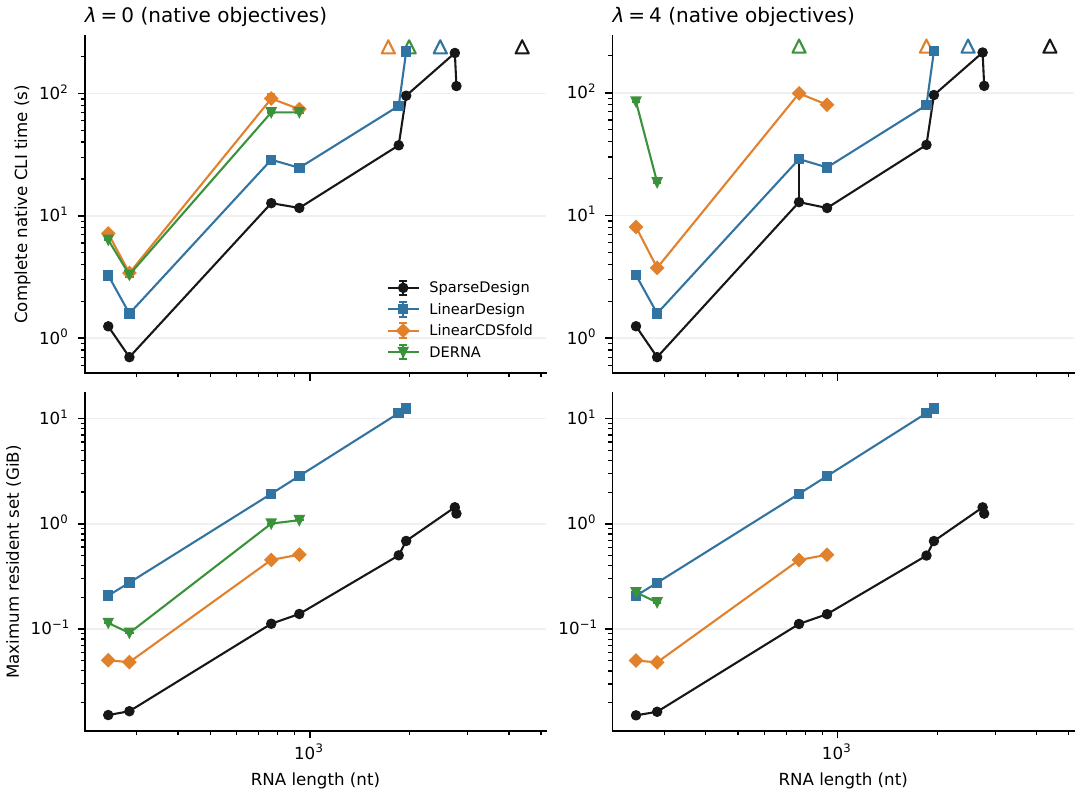}
\caption{Native complete-CLI time and peak RSS. Points show medians and observed ranges
for five fresh successes. Hollow triangles mark \NativeWallCapSeconds-s timeouts; overlapping glyphs have
display offsets without changing RNA-length data. The address-space cap is distinct from measured RSS.}
\label{fig:native-baselines}
\end{figure}

\subsection{Additional large-instance evidence}
The separate July 13 and July 23, 2026 experiments concern human Dp427c
(RefSeq \href{https://www.ncbi.nlm.nih.gov/protein/NP_000100.3}{NP\_000100.3};
3,677 amino acids, 11,031 nucleotides and 12,824 SparseDesign DFA states).
At $\lambda=0$, sparse and the \emph{locally parallelized dense LinearDesign implementation}
return MFE $-7{,}161.40$ kcal/mol; this comparator differs from stock upstream and the
current same-source controls. Runs used one NUMA node of a dual-16-core EPYC 7313 server
with 1 TiB RAM. The artifact retains each protocol's timings and RSS, the single-observation
thread ladder, five separately launched runs and the older \HistAblationPairs-pair ablation.
These observations are not pooled and do not isolate pruning.
An additional counterbalanced packed-layout comparison measured a
\HistPackedSolveIncreasePercent\% median solve-time increase and
\HistPackedRSSReductionPercent\% RSS reduction to \HistPackedRSSGiB{} GiB.
These experiments document attained scale and memory feasibility; the primary EPYC measurements
and separate commodity-PC demonstration are reported above.

\section{How sparse are the candidate lists?}
\label{sec:density}
\subsection{Natural proteins and length-range sensitivity}
Let $Z/n^2$ normalize by squared RNA length and $Z/B_{\rm finite}$ denote retained candidates
relative to all feasible direct state intervals. The latter measures candidate retention;
it does not count how often entries are consumed and is not a runtime prediction.
On the 2,000-protein human-table panel, median $Z/n^2$ is \DensityZero{} at $\lambda=0$ and
\DensityFour{} at $\lambda=4$. Median retention is \RetentionZero\% and \RetentionFour\%,
with maxima \RetentionMaxZero\% and \RetentionMaxFour\%. Every paired natural input retains
fewer candidates at the larger objective weight. These statements are restricted to this
length-balanced panel.

\begin{table}[tb]
\centering\small
\caption{Unweighted log--log candidate-growth fits, $\log Z=\alpha+\beta\log n$.
Changing the minimum included length reveals curvature hidden by the full-range exponent.
Counts and coefficients are generated from the independent frozen-campaign audit.}
\label{tab:cutoff}
\begin{tabular*}{\linewidth}{@{\extracolsep{\fill}}lrrr@{}}
\toprule
Included protein lengths & Proteins & $\beta$, $\lambda=0$ & $\beta$, $\lambda=4$\\
\midrule
Full panel & 2,000 & 1.527 & 1.487 \\
$\geq100$ aa & 1,500 & 1.610 & 1.586 \\
$\geq400$ aa & 1,167 & 1.699 & 1.677 \\
$\geq1{,}000$ aa & 747 & 1.794 & 1.773 \\
\bottomrule
\end{tabular*}

\end{table}

The full-range slopes are \DensityZeroBeta{} and \DensityFourBeta{}, with
$R^2=\DensityZeroRSquared$ and $\DensityFourRSquared$. At lengths of at least 1,000 amino acids,
they rise to approximately 1.79 and 1.77.
Density still declines across the measured long-length range, but the residual pattern and
cutoff sensitivity do not support a single asymptotic exponent. A continuous piecewise fit
with a fixed 400-amino-acid knot reduces the log-space residual sum of squares by
\DensityZeroHingeReductionPercent\% and \DensityFourHingeReductionPercent\%
(\cref{fig:density-sensitivity}); this is a descriptive comparison with
one extra fitted coefficient, not a validated extrapolation model.
To assess dependence on similar sequences, MMseqs2 \citep{Steinegger2017MMseqs2} groups the
2,000 proteins into \DensityClusterCount{} operational similarity clusters
(\DensitySingletonClusters{} singletons; largest \DensityLargestCluster{} sequences),
using at least 30\% aligned-residue identity, 80\% coverage of both sequences and
$E\leq10^{-3}$. These heuristic groups are not asserted biological families.
We use \DensityBootstrapReplicates{} paired resamples, preserving a cluster's multiplicity
across both objective weights and calibrating each resampled length stratum to its original
count. Whole-cluster 95\% bootstrap intervals are
[\DensityZeroClusterCILow, \DensityZeroClusterCIHigh] and
[\DensityFourClusterCILow, \DensityFourClusterCIHigh]. Giving every cluster equal total
weight changes the slopes to \DensityZeroClusterBeta{} and \DensityFourClusterBeta{};
also preserving length-stratum totals gives \DensityZeroCalibratedClusterBeta{} and
\DensityFourCalibratedClusterBeta{}. Thus the observed relation is insensitive to these
weighting choices, while remaining dependent on length range. All intervals quantify
conditional panel sensitivity; they do not establish independent biological sampling.

\begin{figure}[tb]
\centering
\includegraphics[width=\textwidth]{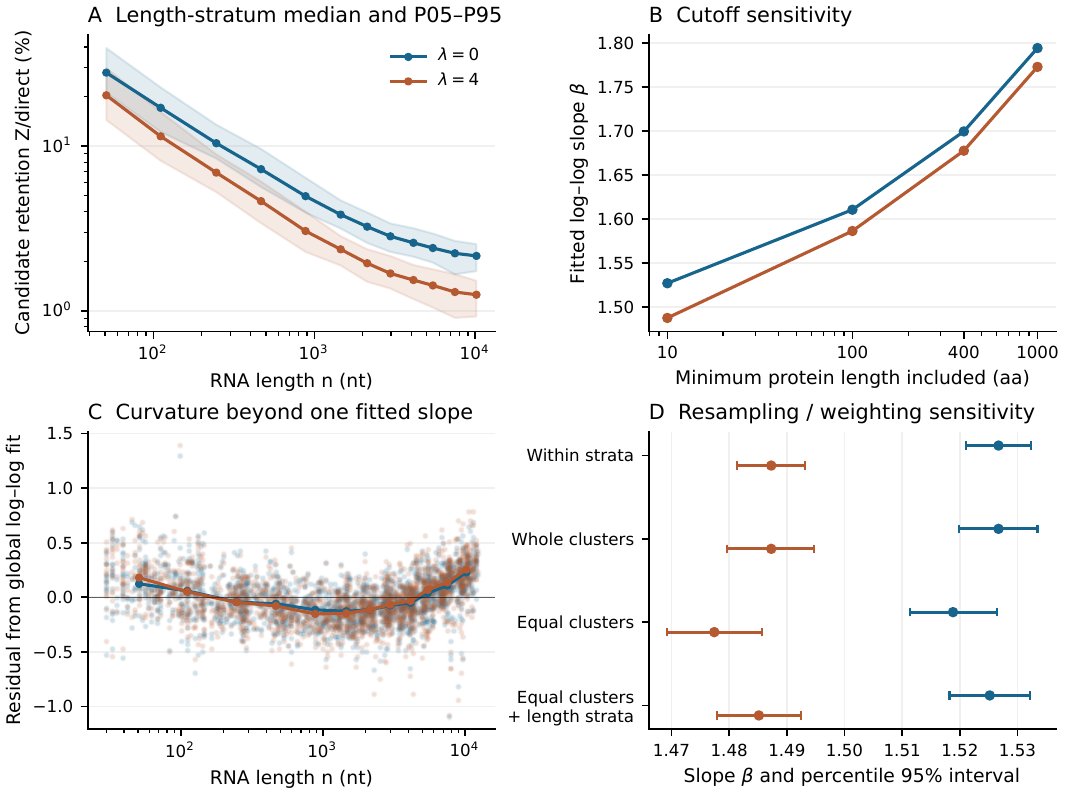}
\caption{Candidate growth on the 2,000-protein panel. Residuals, changing length
cutoffs and cluster-weighted fits qualify the full-range regression: strong contraction
persists, but an exponent near 1.5 is a finite-range summary. The artifact preserves
cluster membership, bootstrap draws, fit tables and all \CorpusRows{} input rows.}
\label{fig:density-sensitivity}
\end{figure}

\subsection{Codon distributions and denser families}
The same 400 proteins are evaluated at $\lambda=4$ under ten codon distributions.
Host-specific slopes span \DensityHostBetaMinimum--\DensityHostBetaMaximum{} and median
$Z/n^2$ spans \DensityHostMedianMinimum--\DensityHostMedianMaximum{}.
Human and yeast tables round synonymous frequencies; the other eight retain full precision
from recorded counts. Differences therefore combine distribution and precision conventions.
Reusing sequences provides a within-input robustness comparison, not ten independent
biological panels. The artifact retains the full host figure and numerical table.

Synthetic families delimit any claim based on natural proteins. At 256 amino acids and
$\lambda=0$, alternating YI retains 146,306 of 186,565 feasible direct intervals (78.42\%).
At $\lambda=4$, only 1,143 remain (0.61\%). At $\lambda=0$, the S family has $Z/n^2\simeq0.389$
throughout 32--256 amino acids, reaching 229,493 candidates at 256 amino acids. The observed growth slopes
are 2.000 for S and 2.026 for YI at $\lambda=0$; YI at $\lambda=4$ instead has slope 1.027
(\cref{fig:validation}). These finite measurements support a useful stress test, not an
asymptotic lower-bound proof. In particular, the natural-panel ``below 44\%'' statement is
not a bound on arbitrary proteins or constrained automata.

\section{Limitations and reproducibility}
\label{sec:limitations}
Only multiloop splits are candidate-sparsified; quadratic tables and dense candidate
families remain. Richer finite-state constraints can enlarge automata. Dangles, coaxial
stacking and pseudoknots require states beyond the scalar dangle-0 interface.
Dense equivalence is proved in real arithmetic;
finite precision, approximate traceback tie recognition and compiler behavior receive
regression tests, not a bit-level proof. Independent checks cannot certify every input
or reported design.

MFE and codon usage are proxies: accessibility and ensemble objectives may rank sequences
differently. No experiment here establishes expression, stability, immune-response or
therapeutic improvements. Finite-range candidate slopes and one successful large instance
establish neither universal scaling nor runtime guarantees. Historical and current protocols differ;
isolating engineering contributions requires appropriate controls.

Code, tests and experimental evidence are at
\mbox{\url{https://github.com/ru-arcl/SparseDesign}}
(release commit \href{https://github.com/ru-arcl/SparseDesign/tree/76618646bcdbd86cc27d04303f9b4d7ae3695f01}{\texttt{7661864}}).
The package includes both layouts, dense controls, profiling, reproduction instructions and
source/data checksums. Supported-envelope and inverse-$\lambda$ tools report supported lines
within stated tolerances and search limits, not every unsupported nondominated design or tied
sequence. Parameter generation pins the complete Turner-2004 text by SHA-256 and checks all
generated headers. Normal builds need neither ViennaRNA nor restricted optimizer binaries;
independent checks list separate dependencies. Original code uses Rutgers' Non-commercial
Research License (RU-NCRL); third-party terms remain separate.

\clearpage
\label{page:references}
\typeout{SPARSEDESIGN-REFERENCE-START-PAGE=\thepage}
\bibliographystyle{unsrtnat}
\bibliography{references}
\end{document}